\documentclass[10pt]{article}
\usepackage{amsthm,amsmath,amssymb}
\usepackage{graphicx} 
\usepackage{algorithm}
\usepackage[noend]{algpseudocode}
\usepackage{natbib}
\usepackage{caption}
\usepackage[export]{adjustbox}
\renewcommand{\cite}{\citep}
\usepackage{titlesec}
\usepackage[margin=1in, footskip=.25in]{geometry}
\usepackage{fancyhdr}
\usepackage[T1]{fontenc}
\usepackage[english]{babel}

\newcommand{\secfnt}{\fontsize{13}{17}}
\newcommand{\ssecfnt}{\fontsize{11}{14}}
\newcommand{\sssecfnt}{\fontsize{10}{14}}

\titleformat{\section}
{\normalfont\secfnt\bfseries}{\thesection}{1em}{}

\titleformat{\subsection}
{\normalfont\ssecfnt\bfseries}{\thesubsection}{1em}{}

\titleformat{\subsubsection}
{\normalfont\sssecfnt\bfseries}{\thesubsubsection}{1em}{}

\titlespacing*{\section} {0pt}{3.5ex plus 1ex minus .2ex}{2.3ex plus .2ex}
\titlespacing*{\subsection} {0pt}{3.25ex plus 1ex minus .2ex}{1.5ex plus .2ex}
\titlespacing*{\subsubsection} {0pt}{3.0ex plus 1ex minus .2ex}{1.0ex plus .2ex}

\newcommand*{\TitleFont}{%
	\usefont{\encodingdefault}{\rmdefault}{b}{n}%
	\fontsize{20}{24}%
	\selectfont}

\newtheorem{theorem}{Theorem}

\newtheorem{problem}{Problem}

\newtheorem{fact}[theorem]{Fact}

\newcommand{\cCW}{{\mathcal {CW}}}
\newcommand{\cX}{{\cal X}}

\def\beq{\begin{equation}}\def\eeq{\end{equation}}

\begin{document}

\title{\TitleFont Scalable Approximation Algorithm for Network Immunization}
\date{}
\author{
	\textbf{Juvaria Tariq}\\ {\small Dept. of Computer Science}	\\ {\small Lahore University of Management Sciences}\\ 14070004@lums.edu.pk
	\\ \\
	\textbf{Imdadullah Khan}\\ 	{\small Dept. of Computer Science}	\\ {\small Lahore University of Management Sciences}\\ imdad.khan@lums.edu.pk
	\and
	\textbf{Muhammad Ahmad}\\	 {\small Dept. of Computer Science}	\\ {\small Lahore University of Management Sciences}\\ 14030004@lums.edu.pk
	\\ \\
	\textbf{Mudassir Shabbir}\\	 {\small Dept. of Computer Science}	\\  {\small Information Technology University, Lahore}	\\ mudassir.shabbir@itu.edu.pk
}





\maketitle

\begin{abstract}
	The problem of identifying important players in a given network is of pivotal importance for viral marketing, public health management, network security and various other fields of social network analysis. In this work we find the most important vertices in a graph $G=(V,E)$ to immunize so as the chances of an epidemic outbreak is minimized. This problem is directly relevant to minimizing the impact of a contagion spread (e.g. flu virus, computer virus and rumor) in a graph (e.g. social network, computer network) with a limited budget (e.g. the number of available vaccines, antivirus software, filters).  It is well known that this problem is computationally intractable (it is NP-hard). In this work we reformulate the problem as a budgeted combinational optimization problem and use techniques from spectral graph theory to design an efficient greedy algorithm to find a subset of vertices to be immunized. We show that our algorithm takes less time compared to the state of the art algorithm. Thus our algorithm is scalable to networks of much larger sizes than best known solutions proposed earlier. We also give analytical bounds on the quality of our algorithm. Furthermore, we evaluate the efficacy of our algorithm on a number of real world networks and demonstrate that the empirical performance of algorithm supplements the theoretical bounds we present, both in terms of approximation guarantees and computational efficiency.
  
\end{abstract}
\vspace{.1in}


\section*{Introduction}
Pairwise interactions between homogeneous entities are commonly modeled as a network. Entities could be computer (computer networks), humans (social network) and electronic components (electricity distribution network) to name a few. Such graphs usually are very large and analytics on them is of pivotal importance among many others, in the field of scientific and engineering research, economics, public health management, and business analytics. In all of these scenarios these systems are threatened by the propagation of harmful entities that moves from a node to its neighbor(s). For instance when a person catches a particular viral disease, the virus will contaminate other people interacting with an already sick person. Similarly, computer virus spread in the whole network with information sharing, and component failure in an electric grid would cause erroneous functionality in the whole network. 

If not contained, this spread of malicious contents will result in an outbreak in the graph with serious effects on the functionality of network. We want to provide some extra capability to some of the nodes in a graph such that those selected nodes will neither be contaminated by malicious content nor will they pass it on to their neighbors. We refer to this as these nodes are immunized. Given a graph topology, our goal is to immunize a subset of nodes that will maximally hinder the spread of undesirable content. As there is cost associated with immunization of nodes, we can only immunize a subset of nodes (not exceeding the given budget). 

We abstractly formulate, this problem, known as the {\em Network Immunization Problem} \cite{chen2016node}, as follows. 
\begin{problem}
	Given an undirected graph $G = (V,E)$, $|V|=n$, and an integer $k<n$, find a set $S$ of $k$ nodes such that `` immunizing'' nodes in $S$, renders $G$ the least ``vulnerable'' to an attack over all choices of $S$ such that $|S|=k$.
\end{problem}
We need to formally define immunizing a node and vulnerability of the graph for a precise formulation of the problem. We use the SIS model of infection spread, where once a node is immunized it remains protected that time on. We also need to quantify graph's vulnerability that is our objective. 

\subsection*{Definitions and Problem Formulation}

For a graph $G=(V,E)$, $A(G)$ denotes the adjacency matrix of $G$ or just $A$ when the graph is clear from the context. For a subset of nodes $S\subset V$, $G_S$ is the subgraph induced by nodes in $S$ and $A_S$ represents its adjacency matrix (i.e. $A(G_S)$). When $S$ is a subset of $V$, $S\subset V$, $G^{[S]}=G_{V\setminus S}$, that is $G^{[S]}$ is the subgraph obtained after removing the nodes of $S$. We denote by $A^{[S]}$ the adjacency matrix of $G^{[S]}$.

$\{\lambda_i(G)\}_{i=1}^n$ or $\{\lambda_i(A)\}_{i=1}^n$ is the eigen spectrum of the graph $G$ or its adjacency matrix $A$. Where $\lambda_{max}(G)=\max_{i} \lambda_i(G)$, is the largest eigenvalue of $A$ (also called spectral radius of $G$) \cite{chung1997}.

For a graph, its epidemic threshold is one of its intrinsic properties. It is of interest to us, as it is well known that if {\em virus strength} is more than the epidemic threshold of the graph, then an outbreak will occur. From the epidemiology literature we get that the epidemic threshold of a graph depends upon the largest eigenvalue of $G$ \cite{chakrabarti2008epidemic}. Hence a common parameter to measure network's vulnerability is the largest eigenvalue of the adjacency matrix of the graph \cite{chen2016node,Ahmad2016}. In this context, our objective reduces to selecting a subset of nodes so as the remaining graph has the as small largest eigenvalue as possible. More precisely, supposed $\lambda_{max}(A^{[S]})$ is the largest eigenvalue of the $G^{[S]}$. Our problem can be formulated as follows: 

\begin{problem}
	\label{problem:1}
	Let $G =(V,E)$ be an undirected graph and let $k$ be an integer $k<|V|$, find a subset of nodes $S\subset V$, with $|S|=k$ such that $\lambda_{max}(A^{[S]})$ is the minimum	possible over all $k$-subsets of $V$.
\end{problem}
In this work we model the Problem \ref{problem:1} as a budgeted combinatorial optimization objective function. Using tools from linear algebra and graph we establish the relationship between the original objective function and the one that we formulate. We define a score of each vertex that is based on the number of closed walks in the graph containing the vertex. We design a randomized approximation algorithm to estimate score of each vertex and greedily select nodes for immunization. Using the fact that our objective function is monotone and sub-modular, we prove a tight analytical guarantee on the quality of our estimate. In addition to theoretical bounds on the quality and runtime of our algorithm we also evaluate our algorithm on various real world graphs. We demonstrate that we achieve up to $100\%$ improvement in terms of drop in vulnerability. Moreover, running time of our algorithm is substantially lower than that of existing solutions. 

\subsection*{Organization}
We provide an outline for the remaining paper. In the following section we provide a detailed background to Problem~\ref{problem:1} and discuss its computational intractability and approaches to approximate it. Our proposed algorithm is presented in the section following that, which also contains approximation guarantees and complexity analysis of our algorithm. The subsequent section contains immunization results from tests of our algorithm on several real world graphs. We also provide performance comparisons of our algorithm with other known algorithms. Section \ref{section:related_work} contains detailed literature review on the problem. A brief conclusion of this work and discussion on future directions is given in the last section.

\section*{Background}\label{section:background}

Eigendrop quantifies the gain after immunizing a set $S$ of $k$ nodes. We can compute eigendrop by $\lambda_{max}(G) -  \lambda_{max}(G^[S])$. Eigendrop depicts how much graph vulnerability has been reduced after immunizing node set $S$. Brute force technique cannot be applied as a solution to problem ~\ref{problem:1} since computing eigenvalue corresponding to ${n \choose k}$ different sets gives $O({n \choose k}\cdot m)$ runtime of the method (largest eigenvalue of a graph can be computed in $O(m)$ \cite{chen2016node})


Indeed, it turns out that solving Problem~\ref{problem:1} optimally is NP-Hard. A straight forward reduction from \textit{Minimum Vertex Cover Problem} follows as,
If there exists a set $S$ with $|S|=k$ such that $\lambda_{max}(A^{[S]})=0$, then $S$ is a vertex cover of the graph. It follows from the following implication of famous \textit{Perron-Frobenius theorem}

\begin{fact}\cite{frobenius}
	Deleting any edge from a simple connected graph $G$ strictly decreases the largest eigenvalue of the corresponding adjacency matrix.
\end{fact}
Also, If there is a vertex cover $S$ of the graph such that $|S|=k$ then deleting $S$ will result in an empty graph which has eigenvalue zero.

Although Problem 1 is NP-Hard, the following greedy algorithm guarantees a $(1-1/e)$-approximation to the optimal solution to Problem~\ref{problem:1}. 
\begin{algorithm}
	\caption{: GREEDY-1($G$,$k$)}
	\label{algo:greedy_eigen}
	\begin{algorithmic}
		\State $S \gets \emptyset$
		\While{$|S|<k$}
		\State $v \gets \underset{x\in V\setminus S}{\arg\max}$ $(\lambda_1(A_{-\{S\cup \{x\} \}}))$
		\State $S \gets S\cup \{v\}$
		\EndWhile
		
		\State \Return $S$ 
	\end{algorithmic}
\end{algorithm}
Approximation guarantee of GREEDY-1 follows from Theorem \ref{NemhauserGreedy}.

\begin{theorem}\cite{Nemhauser}\label{NemhauserGreedy} Let $f$ be a non-negative, monotone and submodular  function, $f: 2^{\Omega} \rightarrow \mathbb{R}$. Suppose ${\cal A}$ is an algorithm, that choose a $k$ elements set $S$ by adding an element $u$ at each step such that  $u= \underset{x\in \Omega\setminus S}{\arg\max}$ $f(S\cup\{x\})$. Then ${\cal A}$ is a $(1-1/e)$-approximate algorithm. 
\end{theorem}

For sparse graphs largest eigenvalue can be computed in $O(m)$, running time of GREEDY-1 amounts to $O(k n m)$. This runtime is impractical for any reasonably large real world graph.

In \cite{chen2016node}, a score, shield-value, was assigned to each subset $S$, which measures the approximated eigendrop achieved by removing the set $S$. They defined a monotone sub-modular function, using which they proposed a greedy algorithm with runtime $O(nk^2 + m)$. Afterwards in \cite{Ahmad2016}, same problem was solved using score based on number of closed walks of length $4$ and a greedy algorithm was presented.

\section*{Our Proposed Algorithm}
As stated in introduction section, we use largest eigenvalue as the measure of vulnerability. In this section we give our approximation algorithm to find the best subset for immunization. We first review some basic facts from linear algebra and graph theory to justify our approach. Let $A$ be an $n\times n$ matrix; the following results from linear algebra [c.f. \cite{strang1988linear}, and \cite{West2001}] relate the eigen spectrum and the trace of $A$.

\begin{fact}\label{traceEigen} 
	$$trace(A)=\sum_{i=1}^n A(i,i)=\sum_{i=1}^{n} \lambda_i(A)$$
\end{fact}

\begin{fact}\label{traceEigenPower} 
	$$trace(A^p)=\sum_{i=1}^{n}\lambda(A^p)=\sum_{i=1}^{n}(\lambda_i(A))^p$$
\end{fact}

From the theory of vector norms \cite{strang1988linear} and fact \ref{traceEigenPower} we know that 
\begin{align*}
&\lim\limits_{\substack{p\rightarrow \infty \\ p \text{ even}}} \left( trace(A^p)\right)^{1/p}  = \lim\limits_{\substack{p\rightarrow \infty \\ p \text{ even}}} \left( \sum\limits_{i=1}^n \lambda_i(A)^p\right)^{1/p} \\
&= \lim\limits_{\substack{p\rightarrow \infty }} \left( \sum\limits_{i=1}^n |\lambda_i(A)|^p\right)^{1/p} = \max_{i}\{\lambda_i(A)\} = \lambda_{max}(A) 
\end{align*}

Using the above relation we establish that for immunization problem, we want to find a set $S$ of vertices in graph $G$ which, when removed, minimizes $trace((A^{[S]})^p)$ where $A^{[S]}$ is the adjacency matrix of $G^{[S]}$. So goodness of a set $S\subset V(G)$ is defined as 
\beq\label{gp} g_p(S)=trace((A^{[S]}))^p \eeq that we want to minimize. Also define complement of the function $g_p(S)$ as
\beq\label{fp} f_p(S)=trace(A^p)-trace((A^{[S]})^p)  \eeq

Clearly minimizing $g_p(S)$ is equivalent to maximizing $f_p(S)$. Now we give combinatorial definition of the optimization functions defined above using the following fact from graph theory. For a set $X\subset V(G)$, and vertices $u,v\in V(G)$, $N_X(v)=N_{G_X}(v)$ is the set of neighbors of $v$ in $X$ and $d_X(v)=d_{G_X}(v)=|N_X(v)|$, called degree of $v$ in $X$. When $X=V(G)$, we refer to $d_G(v)$ as $d(v)$. Moreover $d_X(u,v)$ represents the size of common neighborhood of $u$ and $v$ in set $X$, i.e $d_X(u,v)=|N_X(u)\cap N_X(v)|$. For a vertex $v\in S\subset V(G)$, $\cCW_p(v,S)$ is the set of all closed walks of length $p$ in $G_S$ containing $v$ at least once and $cw_p(v,S)=|\cCW_p(v,S)|$. Similarly we define $\cCW_p(S,G)$ to be the set of closed walks of length $p$ containing vertices of $S$ and correspondingly $cw_p(S,G)$ is the cardinality of the set. For simplicity we write $cw_p(G,G)$ as $cw_p(G)$.

\begin{fact}\cite{West2001}
	Given a graph $G$ with adjacency matrix $A$, $$cw_p(G)=trace(A^p)$$
\end{fact}

From the above fact and definition of trace, we get that 

$$cw_p(G)= cw_p(V\setminus S, G^{[S]})+ cw_p(S,G)$$

Note that this is same equation as \eqref{fp} and can be rewritten as $cw_p(G)=f_p(S)+g_p(S)$. This tells us we need to find set $S$ which maximizes $cw_p(S,G)$ (equivalently $f_p(S)$). Computing $cw_p(S,G)$ is expensive for large value of $p$, but in practice we observe that $p=6$ is sufficiently large.

\begin{theorem}
	Given a graph $G$ with adjacency matrix $A$
	\begin{align*}  cw_6(v,G)&= && 6\sum_{i=1}^n \sum_{j=1}^n A^2(v,v_i)A^2(v,v_j)A^2(v_i,v_j)-3\sum_{i=1}^n \sum_{j=1}^n A^2(v,v_i)A^2(v,v_j)A(v,v_i)A(v,v_j)\\
	& &&-6\sum_{i=1}^n A^2(v,v_i)^2A^2(v,v)+ 2A^2(v,v)^3
	\end{align*}
\end{theorem}

\begin{proof}
	A typical closed walk $W$ of length $6$ can be represented as $(a,b,c,d,e,f,a)$. Note that $v$ can appear in a closed walk of length $6$ at most thrice.

	First we count the walks that contain $v$ exactly once. Lets assume that $v$ appears at the first position, i.e. $W=(v,a,b,c,d,e,v)$. Now since $v$ can not appear at any other position, we get that $b,c,d\neq v$. Also $(v,a,b)$, $(b,c,d)$ and $(d,e,v)$ are paths of length $2$ and number of such walks can be $d(v,b)$, $d(b,d)$ and $d(d,v)$ respectively. But $d(a,d)$ may include $c=v$ case. So to exclude this we subtract $A(b,v)A(d,v)$ from $d(b,d)$ (this will be $1$ only if $v$ is neighbor of both $b$ and $d$). We get that number of closed walks of length $6$ containing $v$ only at first position is $\sum_{b\neq v}\sum_{d\neq v} d(v,b)d(v,d)[d(b,d)-A(v,b)A(v,d)]$. Each one position rotation of this walk results in distinct walk of the kind, so we get that walks containing $v$ exactly once are $6\sum_{b\neq v}\sum_{d\neq v} d(v,b)d(v,d)[d(b,d)-A(v,b)A(v,d)]$.
	
	Now we count the walks containing $v$ twice. One way to get such walk is $v$ is in first and third positions i.e. $W=(v,a,v,b,c,d,v)$. Number of such walks is $\sum_{c\neq v} d(v,c)^2d(v)$. Again each rotation gives unique walk, so we have $6\sum_{c\neq v} d(v,c)^2d(v)$ such walks. Another way to have a walk with $v$ appearing twice is $W=(v,a,b,v,c,d,v)$. There are $\sum_{b\in N(v)}\sum_{d\in N(v)} d(v,b)d(v,d)$ walks with $v$ at first and fourth position which is same as\\ $\sum_{b\in V}\sum_{d\in V} d(v,b)A(v,b)d(v,d)A(v,d)$. Note that only two clockwise rotations result in new walks. This gives that total number of closed walks of length $6$ containing $v$ twice is $$6\sum_{c\neq v} d(v,c)^2d(v)+3\sum_{b\in N(v)}\sum_{d\in N(v)} d(v,b)d(v,d)$$
	
	If we consider walks which contain $v$ thrice, then there are two possibilities for such walks, one which start at $v$ $i)$ $(v,a,v,b,v,c,v)$ and $ii)$ $(a,v,b,v,c,v,a)$. Count for either of them is $d(v)^3$. This gives the total count of these walks as $2d(v)^3$.
	
	So number of closed walks of length $6$ containing a vertex $v$ in graph $G$ is 
	$$6\sum_{b\neq v}\sum_{d\neq v} d(v,b)d(v,d)[d(b,d)-A(v,b)A(v,d)] + 6\sum_{c\neq v} d(v,c)^2d(v)+3\sum_{b\in N(v)}\sum_{d\in N(v)} d(v,b)d(v,d) + 2d(v)^3$$

\end{proof}
%

Clearly computing this number for any vertex $v$ takes $O(n^2+ c(n))$ time where $c(n)$ is the time taken for computing $A^2$. So instead we approximate the number of closed walks of length $6$ containing $v$.
\subsection*{Approximating number of walks}
An equivalent expression for $cw_6(v,G)$ is $$ 6A^6(v,v) - 6A^4(v,v)A^2(v,v)-3(A^3(v,v))^2 +2(A^2(v,v))^3.$$ The formula for $cw(v,G)$ suggests that we need to approximate the powers of adjacency matrix $A$ of $G$. For the purpose, we consider a summary graph $H$ of $G$ which is weighted undirected graph with adjacency matrix $A(H)=C$. We generate matrix $C$ (graph $H$) in the following way:  

First we partition the vertex set $V(G)$ into random subsets using a random \textit{hash function $h$}. Let the partition of vertices under the hash function $h$ be ${\cal P}(h)$ with $|{\cal P}(h)|=\alpha$. Second we construct matrix $C$ as
\begin{algorithm}[H]
	\caption{: SummaryGraph($A(G)$,$\alpha$,$h$)}
	\label{algo:summarygraph}
	\begin{algorithmic}
		\State $C \gets \Call{zeros}{p\times \alpha}$	
		\For{$i=1$ to $n$}
		\For{$j=i$ to $n$}
		
		\If {A(i,j)=1}
		\State $C[h(i)][h(j)] \gets C[h(i)][h(j)]+1$
		\State $C[h(j)][h(i)] \gets C[h(i)][h(j)]$
		\EndIf
		\EndFor
		\EndFor		
		\State \Return $C$ 
\end{algorithmic}\end{algorithm}
This matrix $C$ corresponds to summary graph $H$ in which every node (super-node) represents a set of vertices in ${\cal P}(h)$ and $C(i,j)$ entry denotes the number of edges from super-node $i$ to super-node $j$ (number of edges from vertices in super-node $i$ to vertices in super-node $j$). Lets denote $i$th super-node of $H$ by $\cX_i$

In order to approximate $cw_6(v,G)$ of a vertex $v\in V(G)$, we use powers of matrix $C$ instead those of $A(G)$. We keep $C^2$ and $C^3$ matrices. For each $\cX_i\in V(H)$, we compute terms $C^6(i,i)$ using formula  $\sum_{j=1}^{\alpha} (C^3(i,j))^2$ and $C^4(i,i)$ by $\sum_{j=i}^{\alpha} (C^2(i,j))^2$. 

Note that $C^p(i,j)$ represents the total number of walks of length $p$ from vertices in $\cX_i$ to vertices in $\cX_j$. So we can find the number of closed walks of length $p$ containing a specific vertex $v\in V(G)$, by estimating the contribution of this vertex in the total number of walks in $\cX_{h(v)}$. For this purpose, we define the contribution factor of $v$ with $h(v)=i$ in $C^p(i,i)$ as 

$$\frac{d_{\cX_i}(v)^p}{\sum_{u\in \cX_{i}}d_{\cX_i}(u)^p}$$
This can be seen clear if we expand the terms $C^p(i,i)$. For instance if we expand $C^3(i,i)$, we get one term as $(C(i,i))^3$ which is same as $\left( \sum_{v\in \cX_i} d_{\cX_i}(v)\right)^3 $.
We define $\sum_{u\in \cX_{i}}d_{\cX_i}(u)^p$ to be $D_p(i)$. So we get that the estimated value of $cw_6(v,G)$ is following when $h(v)=i$
$$cw'(v)=6C^6(i,i)\frac{d_G(v)^6}{D_6(i)} - 6d_G(v)C^4(i,i)\frac{d_G(v)^4}{D_4(i)} - 3\left(C^3(i,i)\frac{d_G(v)^3}{D_3(i)} \right)^2 + 2\left( d_G(v)\right)^3 $$

Since we partitioned $V(G)$ using random hash functions, we use multiple hash functions to normalize the effect of randomness as given follow.
\begin{algorithm}[H]
	\caption{: EstimateWalks($A(G)$,$\alpha$,$\beta$)}
	\label{algo:estimatewalks}
	\begin{algorithmic}
		\For {$i=1$ to $\beta$}
			\State $cw_i'\gets \Call{zeros}{n}$
			\State $C_i \gets \Call{SummaryGraph}{A(G),\alpha,h_i}$
			\For{$j=1$ to $n$}
			\State Compute $cw'_i[v_j]$
			\EndFor
		\EndFor
		\State $cwMin\gets \Call{zeros}{n}$
		\For{$j=1$ to $n$}
		\State  $cwMin[v]\gets \min_i cw'_i[v_j]$
		\EndFor
		\State \Return $cwMin$ 
\end{algorithmic}
\end{algorithm}
Once we have estimated the walks for each vertex $v$ of $V(G)$ using multiple hash functions, call it ${W(v)}(v)$, we can select set $S$ for immunization that contain vertices with most number of walks. But for efficient results we would prefer to choose $S$ that have vertices which are well spread apart and we do not want to select a lot of those vertices which are connected to each other. In order to deal with this, we define the score of each candidate subset $S$, on basis of which we select $S$ for immunization. For $v\in V(G)$, and $S\subset V(G)$,

\beq score(S)= \gamma \sum_{v\in S} W(v)^2-\sum_{u,v\in S} W(v)A(u,v) W(u) \eeq

where $\gamma$ is a positive integer. We want to find set $S$ such that $$S= \arg\max score(S), |S|=k.$$ But this optimization problem is clearly computationally intractable since it requires computing score for each of ${n\choose k}$ sets. So we show that the function $score(S)$ is monotonically non-decreasing and sub-modular, allowing us to devise a greedy strategy to construct set $S$ with guaranteed good approximation of our results.

First we show that function $score(S)$ is monotonically non-decreasing function. Let $E,F\subset V(G)$ and $x\in V(G)$ with $F=E\cup \{x\}$. Consider
\begin{align*}
score(F)-score(E)&= \gamma \sum_{v\in F} W(v)^2-\sum_{u,v\in F} W(v)A(u,v) W(u) -\gamma \sum_{v\in E} W(v)^2 +\sum_{u,v\in E} W(v)A(u,v) W(u) \\
&= \gamma W(x)^2-\sum_{v\in E} W(v)A(x,v) W(x)\\
&= W(x)\left[\gamma W(x)-\sum_{v\in E} W(v)A(u,v) \right]\\
&\geq 0 
\end{align*}
Now since $\gamma$ is any positive integer, if we keep $\gamma\geq k \max_{v\in V(G)} \{W(v)\}$, last inequality is satisfied and hence $score$ function is monotonically non-decreasing. Let $I,J,K\subset V(G)$ with $I\subset J$.

Now we prove the sub-modularity of this function. 
\begin{align*}
& { }(score(I\cup K)-score(I))-(score(J\cup K)-score(J))\\
&= \left(\gamma \sum_{v\in I\cup K} W(v)^2-\sum_{u,v\in I\cup K} W(v)A(u,v) W(u)-\gamma \sum_{v\in I} W(v)^2+\sum_{u,v\in I} W(v)A(u,v) W(u)   \right) \\
&-\left(\gamma \sum_{v\in J\cup K} W(v)^2-\sum_{u,v\in J\cup K} W(v)A(u,v) W(u)-\gamma \sum_{v\in J} W(v)^2+\sum_{u,v\in J} W(v)A(u,v) W(u)   \right) \\
&=\left( \gamma \sum_{v\in K} W(v)^2-\sum_{u,v\in K} W(v)A(u,v) W(u) -2\sum_{u\in K, v\in I} W(v)A(u,v) W(u)\right) \\
&-\left(\gamma \sum_{v\in K} W(v)^2-\sum_{u,v\in K} W(v)A(u,v) W(u) -2\sum_{u\in K, v\in J} W(v)A(u,v) W(u) \right) \\
&= 2\sum_{u\in K, v\in J} W(v)A(u,v) W(u) -2\sum_{u\in K, v\in I} W(v)A(u,v) W(u)= 2\sum_{u\in K, v\in J\setminus I} W(v)A(u,v) W(u)\geq 0
\end{align*}
Proving that our optimization function is sub-modular, and we can use Theorem \ref{NemhauserGreedy}, which guarantees that the greedy strategy will be $(1-1/e)$-approximate algorithm. We give the following greedy algorithm to construct the required set $S$.
\begin{algorithm}[H]
	\caption{: GreedyNodeImmunization($A(G)$,$k$,$\alpha$,$\beta$)}
	\label{algo:greedyalgo}
	\begin{algorithmic}[1]
		\State $S \gets \emptyset$
		\State $W_2, Score \gets \Call{zeros}{n} $
		\State $W \gets \Call{EstimateWalks}{A(G),\alpha,\beta}$
		\State $\gamma \gets \max_{i} W[i]$
		
		\For{$i=1$ to $n$}
			\State $W_2[i]\gets  \gamma W[i]^2$
		\EndFor
		
		\For{$i=1$ to $k$}
			\State $a_S \gets A[:, S]*W[S] $
			\For{$j=1$ to $n$}
				\If{$j\notin S$}
					\State $Score[j] \gets W_2[j]-2a_S[j]W[j] $
				\Else
					\State $Score[j] \gets -1$
				\EndIf
			\EndFor
			\State $maxNode \gets \arg \max_{j} Score[j] $
			\State $S \gets S\cup\{maxNode\}$
		\EndFor
		\State \Return $S$
\end{algorithmic}
\end{algorithm}
\subsubsection*{Analysis of Algorithm}
 Now we analyze our proposed algorithm and give its runtime complexity. First we discuss complexity of EstimateWalks function. This function needs to compute the following $\beta$ times (count of hash functions): $C$, $C^2, C^3, C^4, C^6$, $D_6(i)$ for all sets in partition formed by hash function, $cw'(v)$ for $n$ vertices.

Note that $C$ matrix for all hash functions can be computed with one scan of the whole graph, which takes $n^2$ time. Computing all the above, except $C$ takes at most $O(\alpha^3)$ time. For every hash function, it takes maximum $O(n+\alpha^3)$ time and finding min $cw'(i)$ for each vertex takes $\beta n$ time. This implies $EstimateWalks$ function takes $O(n^2+ \beta(n+\alpha^3))$ time.

Line $4$ and first for loop takes $O(n)$ steps. The $j$th iteration if loop in lines $9$ to $13$ takes $O(n+nj)$ and line $14$ is $O(n)$ work. This shows that second loop from line $7$ to $15$ takes $\sum_{j=1}^k O(n+nj)$ time which is $O(nk^2)$ in total.

So GreedyNodeImmunization($A(G)$, $k$, $\alpha$,$\beta$) algorithm takes total $O(n^2+ \beta(n+\alpha^3)+nk^2)$ time.

\section*{Experiments}
We present results of our suggested algorithm in detail in this section. We have compared results of our algorithm with those of NET-SHEILD\footnote{https://www.dropbox.com/s/aaq5ly4mcxhijmg/Netshieldplus.tar}, Brute Force Method and Walk $4$\cite{Ahmad2016} to evaluate the quality and efficiency. NET-SHEILD selects the vertices based on the eigen vector corresponding to largest eigenvalue of graph, Brute Force algorithm picks vertices which have maximum number of closed walks of length six passing across them and Walk $4$ chooses nodes based on  approximation of walks of length $4$ for immunization purpose. We have implemented the algorithm in Matlab and we have made our code available at the given link. 
\begin{table}[h!]
	\centering
	\begin{tabular}{ |p{2cm}|p{2cm}|p{2cm}| }
		
		\hline
		\textbf{Name} & \textbf{Nodes} \#  & \textbf{Edges} \# \\
		\hline
		Karate & 34 & 78 \\
		\hline
		Oregon & 10,670 & 22,002 \\
		\hline
		AA & 418,236 & 2,753,798\\
		\hline
		
	\end{tabular}
	\caption{Summary of Datasets}
	\label{tableOne}
\end{table}
We have used real world graphs for experimentation and all our graphs are undirected and unweighted. The first data set called Karate graph\footnote{http://konect.uni-koblenz.de/networks/ucidata-zachary} is a small graph of local karate club in which nodes represent members of the club and an edge between two nodes shows friendship among corresponding members. Karate graph consists of 34 nodes and 78 edges.\\
Second dataset is obtained from Oregon AS (Autonomous System)\footnote{http://snap.stanford.edu/data/oregon1.html} router graphs. We have constructed a communication graph in which nodes are participating routers and an edge between two routers represents direct peering relationship among them. A number of Oregon graphs are available and each graph is made from communication log of one week. We have selected a graph containing 10,670 nodes and 22,002 edges.\\
The third data set (AA) is from DBLP\footnote{http://dblp.uni-trier.de/xml/} dataset. In this graph a node represents an author and presence of an edge between two nodes shows that two authors have a co-authorship. In DBLP there is total node count of 418,236 and the number of edges among nodes is 2,753,798. We extracted smaller sub-graphs by selecting co-authorship graphs of individual journals (e.g Displays, International Journal of Computational Intelligence and Applications, International Journal of Internet and Enterprise Management, etc.). We ran our experiments on 20 different smaller co-authorship graphs of different journals. For the smaller sub graphs that we have extracted from DBLP dataset, node count goes up to few thousands and edge count goes up to few ten thousands. Details of sub graphs of DBLP data set is given in Table \ref{table2}. These subgraphs are also undirected and unweighted.
\begin{table}[h!]
	\centering
	\begin{tabular}{ |p{5.3cm}|p{1.3cm}|p{1.3cm}| }
		
		\hline
		\textbf{Name} & \textbf{Nodes}  & \textbf{Edges} \\
		\hline
		AI Communication & 1,203 & 2,204 \\
		\hline
		APJOR & 1,132 & 1,145 \\
		\hline
		Computer In Industry & 2,844 & 4,466\\
		\hline
		Computing And Informatics (CAI) & 1,598 & 2,324\\
		\hline
		Decision Support Systems (DSS) & 4,926 & 14,660\\
		\hline
		Display & 1,374 & 3,204\\
		\hline		
		Ecological Informatics & 1,990 & 4,913 \\
		\hline
		Engineering Application of AI & 4,164 & 6,733\\
		\hline
		IJCIA & 848 & 975\\
		\hline
	\end{tabular}
	\caption{Summary of DBLP subgraphs}
	\label{table2}
\end{table}

\begin{figure}[h!]
\centering
\begin{minipage}{.5\textwidth}

  \includegraphics[valign=c,width=9cm,height=9cm,keepaspectratio]{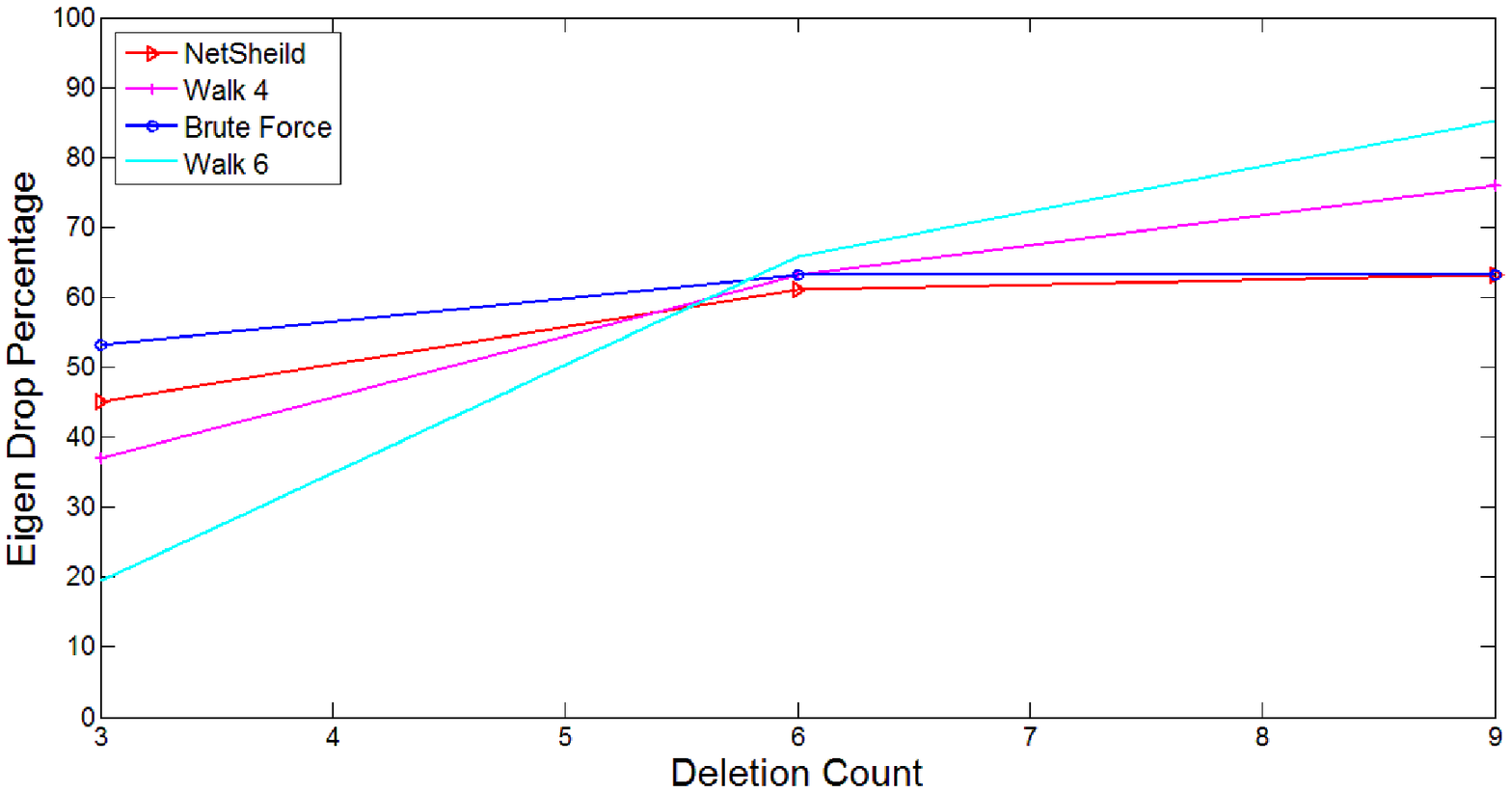}
  \caption{Eigendrop of Karate Graph}
  \label{fig:test1}
\end{minipage}%
\centering
\begin{minipage}{.5\textwidth}
  \includegraphics[valign=c,width=9cm,height=9cm,keepaspectratio]{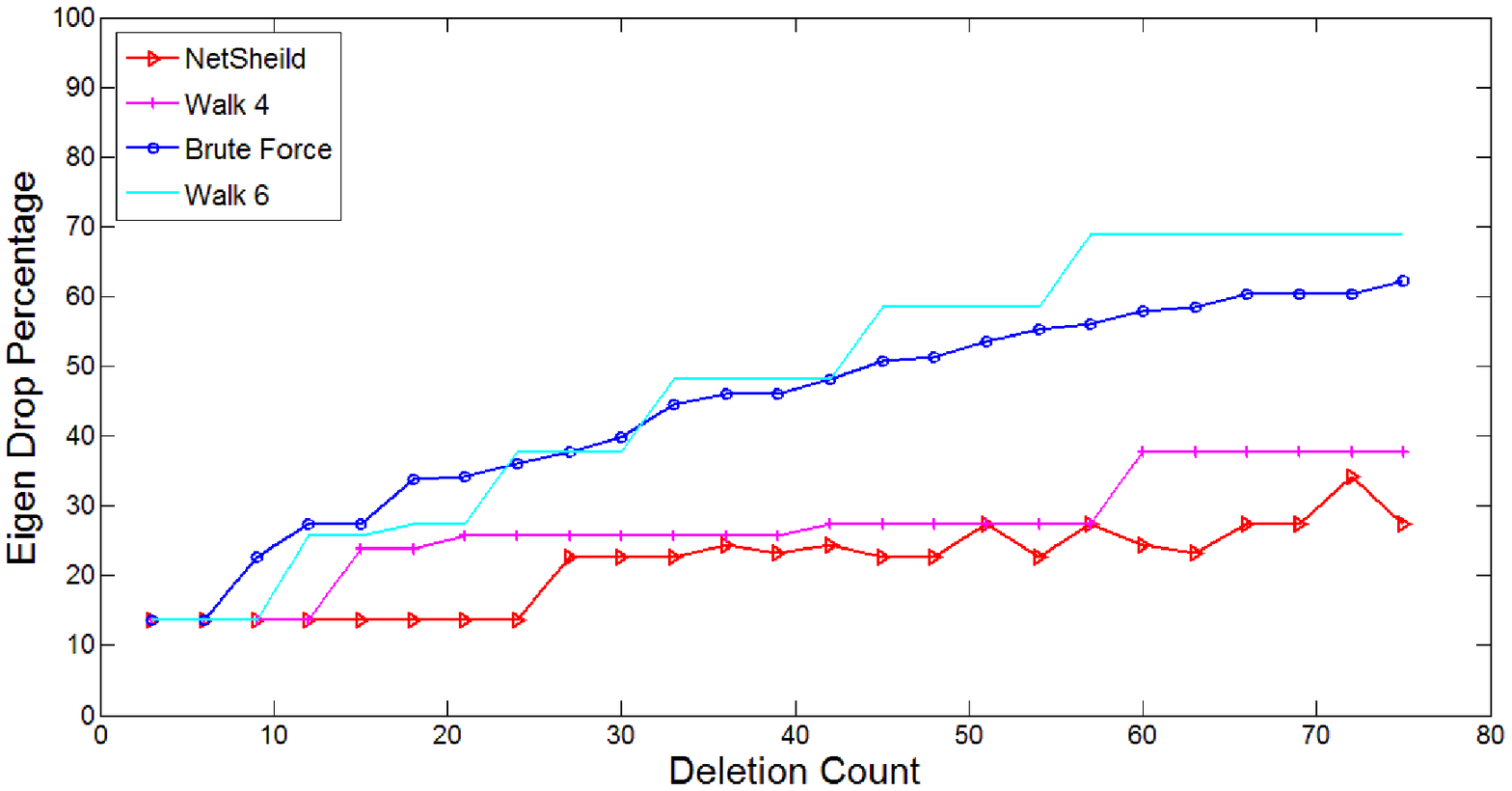}
	\caption{Eigendrop of Computer In Industry Graph}
  \label{fig:test2}
\end{minipage}
\end{figure}

\begin{figure}[h!]
\begin{minipage}{.5\textwidth}
  \centering
  \includegraphics[width=9cm,height=9cm,keepaspectratio]{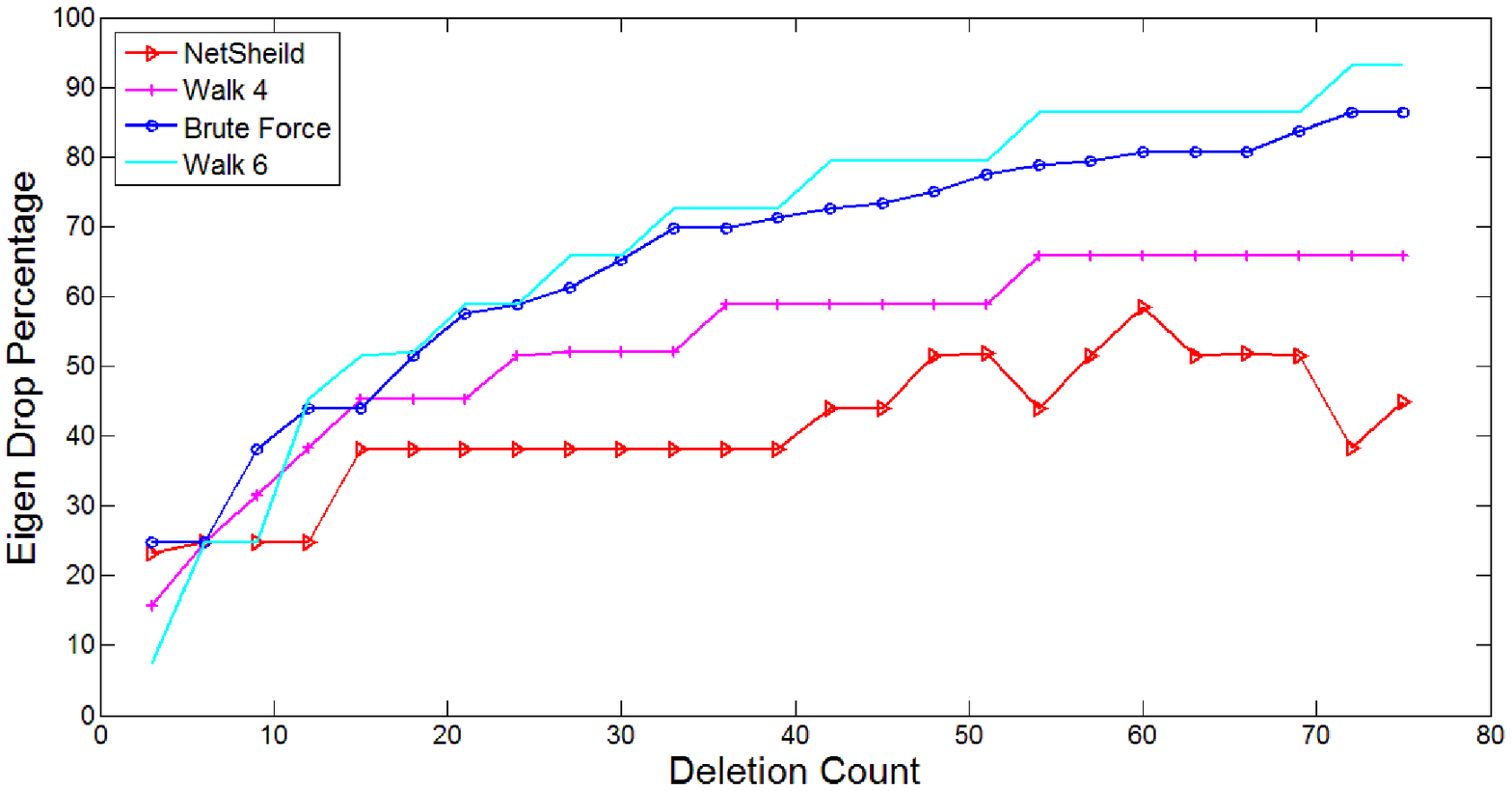}
	\caption{Eigendrop of CAI Graph}
  \label{fig:test3}
\end{minipage}%
\hspace{0.05in}
\begin{minipage}{.5\textwidth}
  \centering
  \includegraphics[width=9cm,height=9cm,keepaspectratio]{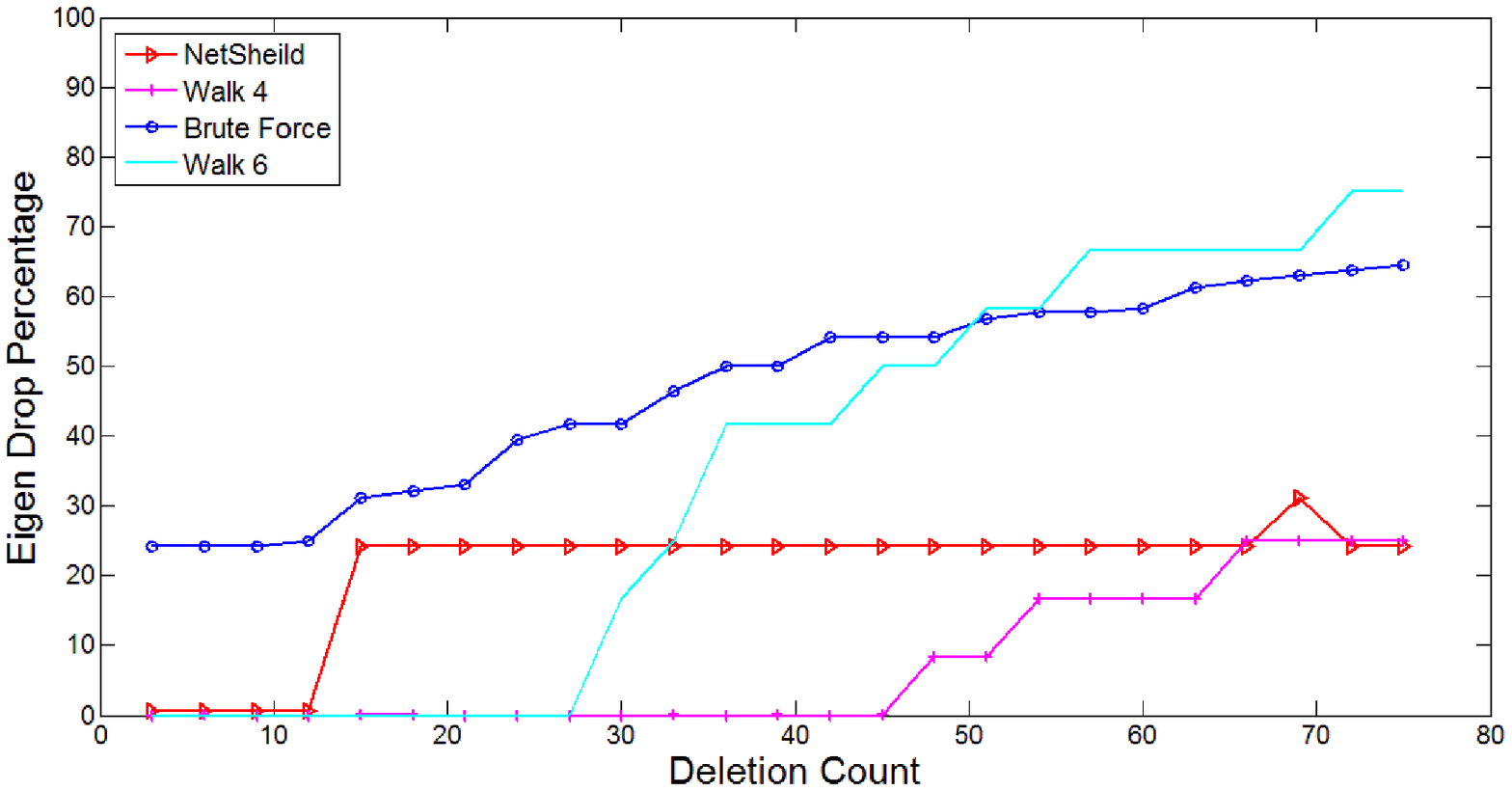}
	\caption{Eigendrop of DSS Graph}
  \label{fig:test4}
\end{minipage}
\end{figure}

\begin{figure}[h!]
\begin{minipage}{.5\textwidth}
  \centering
  \includegraphics[width=9cm,height=9cm,keepaspectratio]{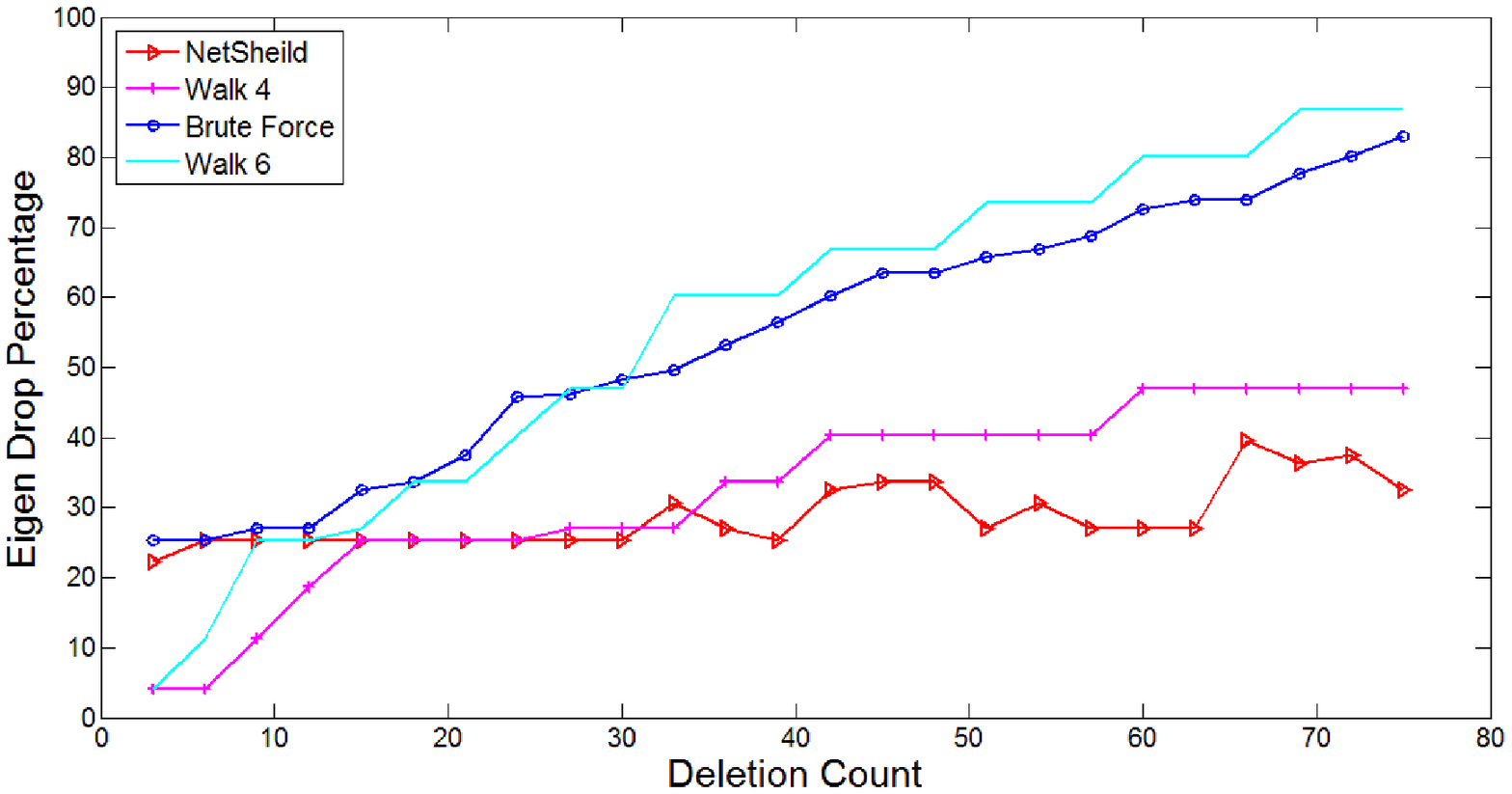}
	\caption{Eigendrop of Displays Graph}
  \label{fig:test3}
\end{minipage}%
\begin{minipage}{.5\textwidth}
  \centering
  \includegraphics[width=9cm,height=9cm,keepaspectratio]{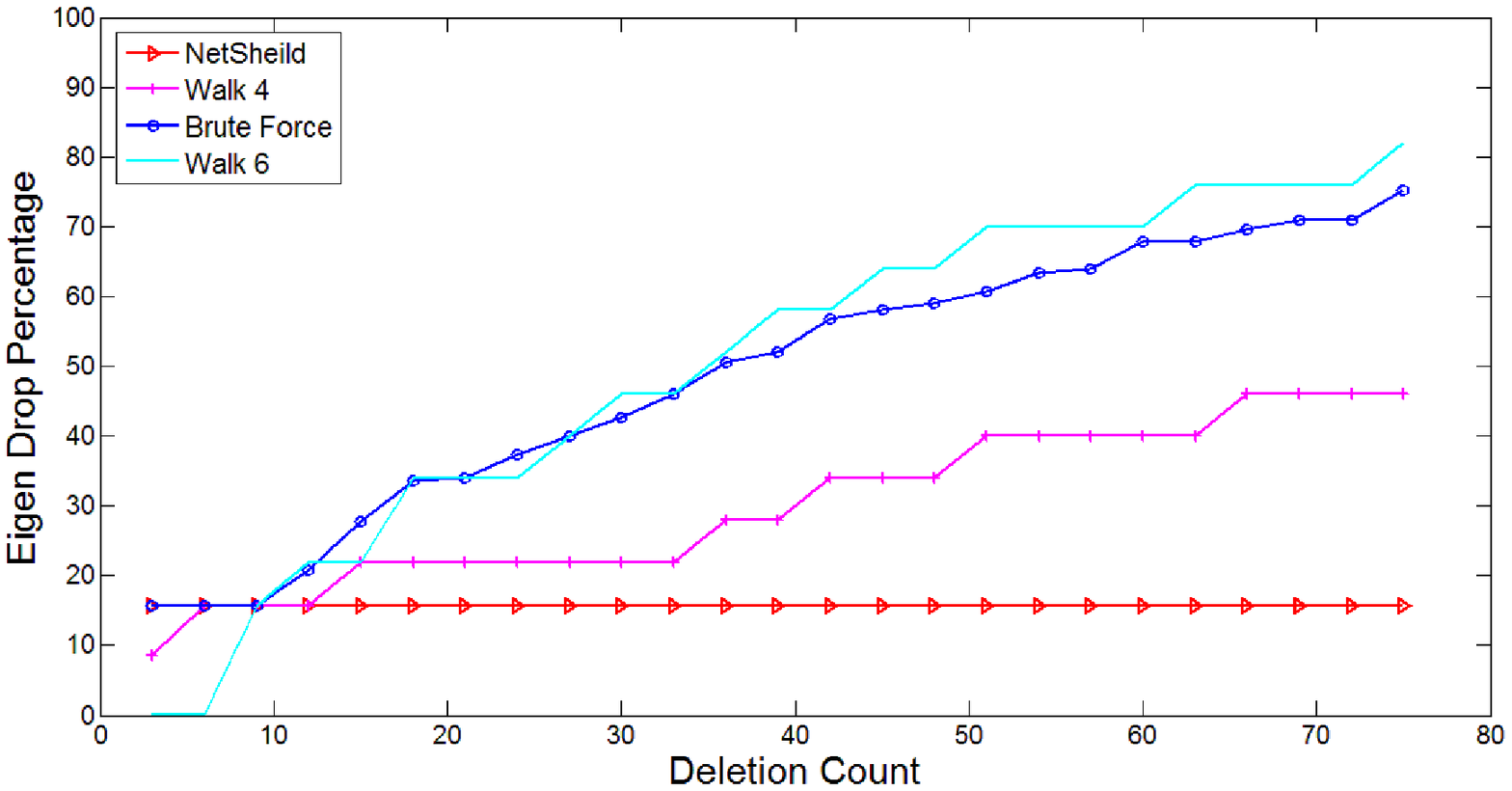}
	\caption{Eigendrop of Ecological Informatics Graph}

  \label{fig:test4}
\end{minipage}
\end{figure}

\begin{figure}[h!]
\begin{minipage}{.5\textwidth}
  \centering
  \includegraphics[width=9cm,height=9cm,keepaspectratio]{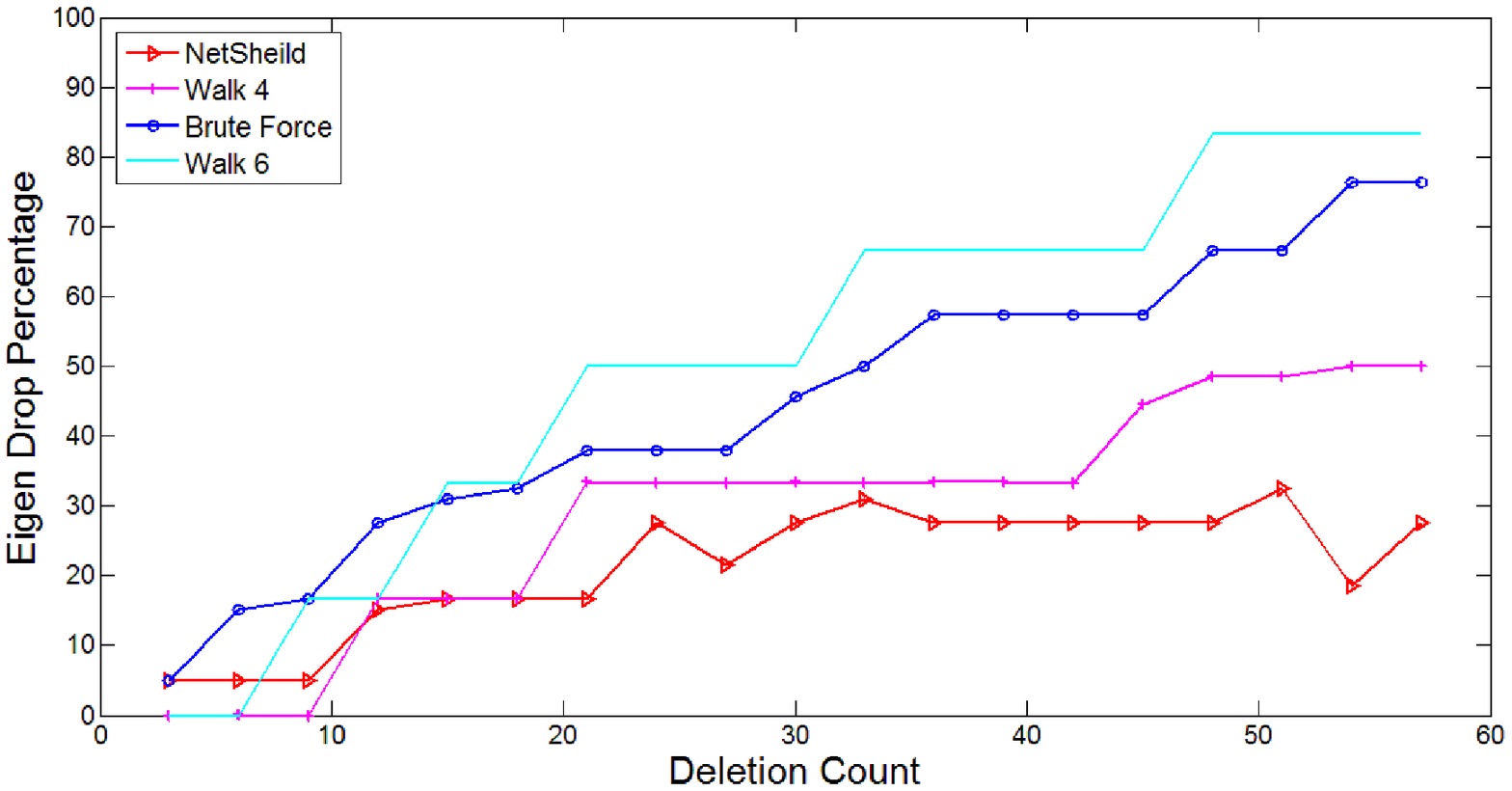}
	\caption{Eigendrop of IJCIA Graph}
  \label{fig:test3}
\end{minipage}%
\begin{minipage}{.5\textwidth}
  \centering
  \includegraphics[width=9cm,height=9cm,keepaspectratio]{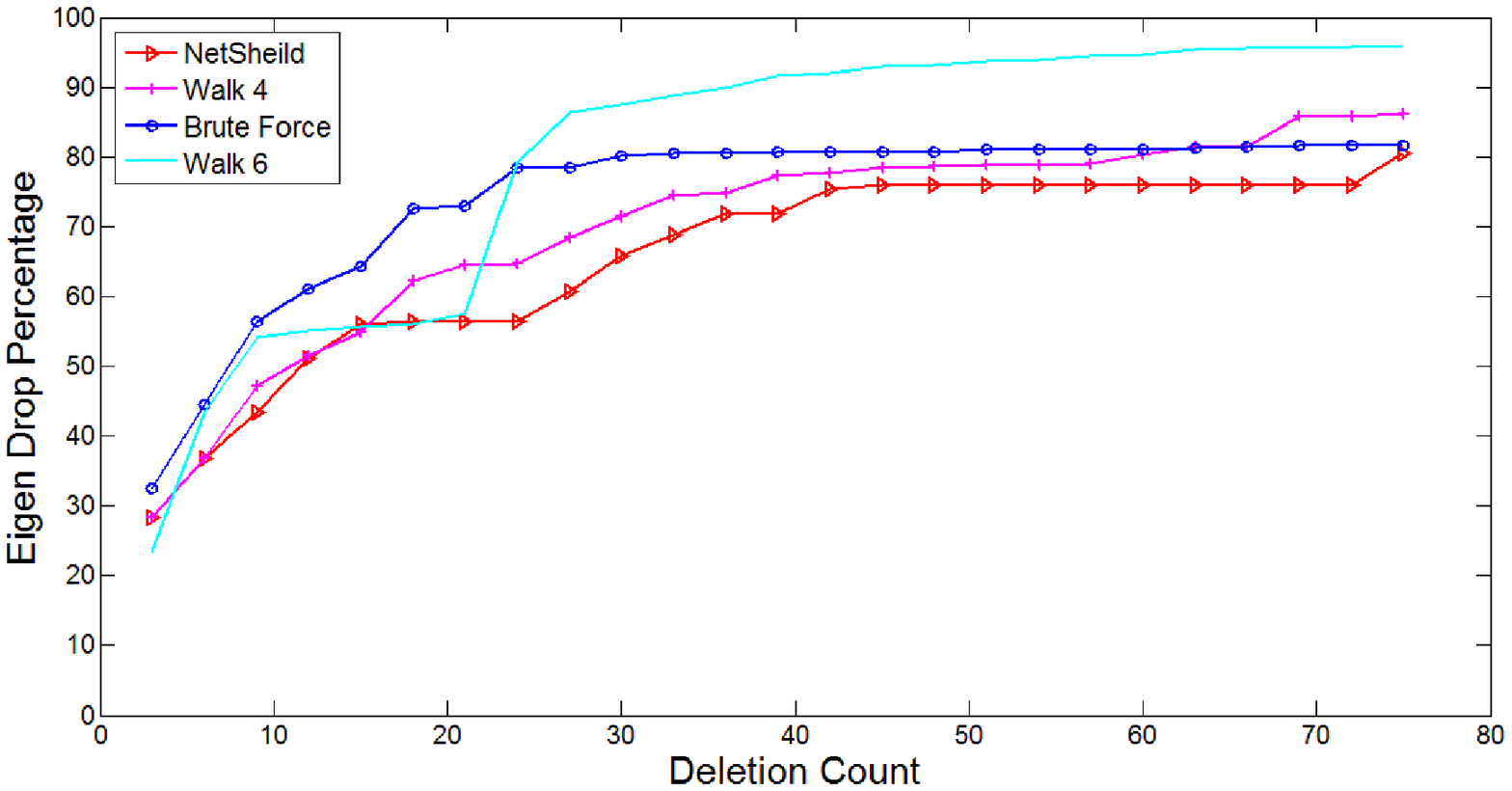}
	\caption{Eigendrop of Oregon Graph}

  \label{fig:test4}
\end{minipage}
\end{figure}

We performed extensive experimentation with varying number of nodes to be immunized in graph. In the results shown, x-axis shows the count of nodes being immunized denoted by $k$ while y-axis shows the benefit achieved in terms of percentage of eigen drop after immunizing $k$ nodes in graph. It is clear from results that our algorithm beats other variants for immunizing the graph in terms of effectiveness. Our algorithm has less computational cost than its competitors and is scalable for larger values of $k$ and also for large size graphs. It is worth mentioning that our algorithm achieves high accuracy in terms of approximation with very small $k$. Hence for very large graphs, our algorithm will achieve reasonable level of accuracy in very little time.

\section*{Related work}\label{section:related_work}

In this section we present a review of related work that has been done to target node immunization problem. A vast amount of work has been done to approach this problem using dimensions of spectral graph techniques, information diffusion and selection of central nodes in graph etc. In 2003 Brieseneister, Lincoln and Porras \cite{briesemeister2003epidemic} studied the propagation styles and infection strategies of viruses in communication networks to target susceptible nodes. They aim to do analysis of graphs to make them more defensible against infection. Along with this, the effects of graph topology in the spread of an epidemic are described by Ganesh, Massouli\'{e} and Towsley in \cite{ganesh2005effect} and they discuss the conditions under which an epidemic will eventually die out. Similarly Chakrabarti et. al in \cite{chakrabarti2008epidemic} devise a nonlinear dynamical system (NLDS) to model virus propagation in communication networks. They use the idea of \textit{birth rate, $\beta$, death rate,$\delta$, and epidemic threshold,$\tau$, } for a virus attack where {birth rate} is the rate with which infection propagates, {death rate} is the node curing rate and {epidemic threshold} is a value such that if $\beta / \delta \textless \tau$, infection will die out quickly else if  $\beta / \delta \textgreater \tau$ infection will survive and will result in an epidemic. For undirected graphs, they prove that epidemic threshold $\tau$ equals 1/$\lambda$ where $\lambda$ is largest eigenvalue of adjacency matrix $A$ of the graph. Thus for a given undirected graph, if $\beta / \delta \textless 1/\lambda$, then the epidemic will die out eventually.

The problem has also been addressed through edge manipulation schemes. In \cite{kuhlman2013blocking} dynamical systems are used to delete appropriate edges to minimize contagion spread. While Tong et al. in \cite{tong2012gelling} use the edge removal technique to protect a graph from outside contagion. They remove $k$ edges from the graph to maximize the eigendrop (difference in largest eigenvalues of original and resultant graphs) by selecting edges on the basis of corresponding left and right eigenvectors of leading eigenvalue of the graph such that for each edge $e_x$, score($e_x$) is the dot product of the left and right eigenvectors of leading eigenvalue of adjacency matrix of A.

Graph vulnerability is defined as measure of how much a graph is likely to be affected by a virus attack. As in \cite{tong2012gelling}, the largest eigenvalue of adjacency matrix is selected as a measure of graph vulnerability, in \cite{chen2016node} they also use largest eigenvalue for the purpose but instead of removing edges, nodes are deleted to maximize the eigendrop. Undirected, unweighted graphs are considered and nodes are selected by an approximation scheme using the eigenvector corresponding to largest eigenvalue which cause the maximum eigendrop. 

Probabilistic methods are also used for node immunization problem. Zhang et al. and Song et al. adapt the non-preemptive strategy i.e. selection of nodes for immunization is done after the virus starts propagating across the graph. For this they use discrete time model to obtain additional information of infected and healthy nodes at each time stamp. In \cite{song2015node} directed and weighted graphs are used in which weights represent the probability of a healthy node being contaminated by its affected neighbors and node selection is done on the basis of these probabilities. Then results are evaluated on the basis of save ratio (SR) which is the ratio between the number of infected nodes when k nodes are immunized over the number of infected nodes with no immunization. The work in \cite{zhang2014dava} and \cite{zhang2014scalable} considers undirected graphs and constructs dominator trees for selecting nodes. Results are evaluated in terms of expected number of remaining infected nodes in the graph after the process of immunization.

Other important and closely related problem is $k$ facility location and a lot of work is done on this. In filter placement \cite{erdos2012filter}, those nodes are identified whose deletion will maximally reduce information multiplicity in graph and node selection is done on the basis of number of paths passing through it. Moreover some reverse engineering techniques are also used for similar problems to find out the initial culprits of infection propagation. Prakash, Vreeken and Faloutsos \cite{prakash2012spotting} study the graphs in which virus has already spread for some time and they point out those nodes from where the spread started. From this they find out the likelihood of other nodes being affected.

Another direction to look at the problem is to consider graphs in which some nodes are already infected and these nodes can spread virus among other reachable nodes or graphs in which all nodes are contaminated and the goal is to decontaminate the graph by using some agent nodes which traverse along the edges of the graph and clean the nodes. The problem is usually referred to as decontamination of graph or graph searching problem. Different models are studied to solve the problem and most of them assume the monotonicity in decontamination i.e once a node is decontaminated then it cannot get contaminated again \cite{bienstockseymour1991},\cite{flocchini2008},\cite{flocchini2007},\cite{fraigniaudnisse2008}. But non-monotonic strategies are also studied \cite{daadaa2016network}.

Other work that is related to node immunization is the selection of most influential nodes in a given network to maximize the information diffusion in a network. Kempe et al. provided the provably efficient approximation algorithm for the problem \cite{kempe2003}. Seeman and Singer \cite{seeman2013adaptive} use stochastic optimization models to maximize the information diffusion in social networks. Influence maximization problem is slightly different from immunization problem as in influence maximization problem the goal is to select nodes for seeding which will maximize the spread on new idea while in node immunization problem the aim is to select nodes which will help in minimal spread of virus.

\section*{Conclusion}

In this work, we explored some links between established graph vulnerability measure and other spectral properties of even powers of adjacency matrix of the graph. We define shield value in terms of trace of the adjacency matrix of the graph. Based on these insights we present a greedy algorithm that iteratively selects $k$ nodes such that the impact of each node is maximum in the graph, in the respective iteration, and thus we maximally reduce the spread of a potential infection in the graph by removing those vertices. Our algorithm is scalable to large. We have done experimentation on different real world communication graphs to prove the accuracy and efficiency of our algorithm. Our algorithm beats the state of the art algorithms in performance as well as in quality.

For the future work, we consider the larger values and generalized even parameter $k$, used for our shield value. Hence, we will aim to effectively improve the quality of the estimates. Techniques like locality sensitive hashing can be incorporated for the efficient approximation of the shield value for general $k$.

\bibliographystyle{agsm}    


\bibliography{Improved_bib}

\end{document}